\documentclass[a4paper,11pt]{article}

\usepackage{tikz}
\usetikzlibrary{chains,fit,shapes,arrows}

\usepackage{paralgorithmic}
\usepackage{algorithm}
\usepackage{latexsym}

\usepackage{fullpage}
\usepackage{color}
\usepackage{graphicx}
\usepackage{listings}

\newtheorem{invariant}{Invariant}
\newtheorem{lemma}{Lemma}
\newtheorem{proposition}{Proposition}
\newenvironment{proof}{\emph{Proof:}}{$\Box$\newline}

\newcommand{\arcg}[2]{\langle #1,#2\rangle}

\newcommand{\ADJ}{\textrm{ADJ}}
\newcommand{\indeg}{\textsf{indeg}}
\newcommand{\outdeg}{\textsf{outdeg}}
\newcommand{\inv}{\textsf{in}}
\newcommand{\outv}{\textsf{out}}
\newcommand{\nexti}{\textsf{next}}
\newcommand{\previ}{\textsf{prev}}
\newcommand{\firsti}{\textsf{first}}
\newcommand{\dfsnum}{\textsf{traversal}}
\newcommand{\bfsnum}{\textsf{traversal}}
\newcommand{\bfslev}{\textsf{distance}}
\newcommand{\parent}{\textsf{parent}}

\newcommand{\enque}{\textsf{enque}}
\newcommand{\deque}{\textsf{deque}}
\newcommand{\eliminate}{\textsf{eliminate}}

\title{A Note on (Parallel) Depth- and Breadth-First Search by Arc Elimination}
\author{Jesper Larsson Tr\"aff\\
Vienna University of Technology\\
Faculty of Informatics, Institute of Information Systems\\ 
Research Group Parallel Computing\\
Favoritenstrasse 16/184-5, 1040 Vienna, Austria\\
email: \texttt{traff@par.tuwien.ac.at}}
\date{May 6th, 2013}
\begin{document}
\maketitle

\begin{abstract}
This note recapitulates an algorithmic observation for ordered
Depth-First Search (DFS) in directed graphs that immediately leads to
a parallel algorithm with linear speed-up for a range of processors
for non-sparse graphs. The note extends the approach to ordered
Breadth-First Search (BFS).  With $p$ processors, both DFS and BFS
algorithms run in $O(m/p+n)$ time steps on a shared-memory parallel
machine allowing concurrent reading of locations, e.g., a CREW PRAM,
and have linear speed-up for $p\leq m/n$. Both algorithms need $n$
synchronization steps.
\end{abstract}

\section{Introduction}

Depth- and Breadth-First Search are elementary graph traversal
procedures with simple, sequential
algorithms~\cite{CormenLeisersonRivestStein09,Tarjan72,Tarjan83}. Both
procedures pose problems for parallel implementation: the ordered
Depth-First Search (DFS) problem is $P$-complete~\cite{Reif85}, and
therefore unlikely to admit polylogarithmically fast, parallel
algorithm using only polynomial resources, whereas for Breadth-First
Search (BFS), no work-optimal, polylogarithmically fast parallel
algorithm is known. This note re-presents the simple, work-optimal,
linear time, parallel algorithm for Depth-First Search by Varman and
Doshi~\cite{VarmanDoshi87} (also Vishkin, personal communication, see
also~\cite{EdwardsVishkin12} which describes a number of further
applications), that can give linear speed-up for graphs that are not
too sparse, and extends the basic observation to Breadth-First Search
where an algorithm with similar properties is given. The idea is
simple: instead of examing arcs in the ``forwards'' direction, as in
standard, textbook formulations of DFS and
BFS~\cite{CormenLeisersonRivestStein09}, incoming, ``backwards'' arcs
are used to eliminate arcs that are no longer relevant for the
search. Whereas the standard algorithms have either conflicts (BFS)
and/or dependencies (DFS) that hamper parallelization, arc elimination
can be performed fully in parallel.

Let $G=(V,E)$ be a directed graph with $n=|V|$ vertices and $m=|E|$
arcs (directed edges). Vertices are assumed to be numbered
consecutively, such that $V=\{0,\ldots,n-1\}$. Arcs are ordered pairs
of vertices with $\arcg{u}{v}, u,v\in V$ denoting the arc directed
from source $u$ to target $v$.

It will be assumed that the input graph $G$ is given as an $n$-element
array $\ADJ$ of adjacency arrays. For each vertex $u\in V$,
$\ADJ[u].\outdeg$ stores the out-degree of $u$, and the target vertex
$v_i$ of the $i$th arc $\arcg{u}{v_i}$ for $0\leq i<\ADJ[u].\outdeg$ is
stored in $\ADJ[u].\outv[i]$.

Depth- and Breadth-First Search are procedures for graph traversal
starting from a given a start vertex $s\in V$. Both procedures assign
traversal numbers to the vertices, indicating the order in which they
are reached. Breadth-First Search additionally computes for each
vertex its distance (shortest path in number of traversed arcs) from
the start vertex. Both procedures also compute the search tree, which
will be represented by a parent pointer. For each vertex $u\in
V$, these computed values will be stored in $\ADJ[u].\dfsnum$,
$\ADJ[u].\bfslev$ and $\ADJ[u].\parent$ (for $u\neq s$), respectively.

The search procedures will modify the input graph by eliminating arcs,
and both maintain the following invariant.

\begin{invariant}
\label{inv:eliminated}
A vertex $v$ is called \emph{visited} when it has been assigned its
(Depth- or Breadth-First Search) traversal number.
Each visited vertex $v\in V$ will have no incoming arcs, that is,
there will be no arc $\arcg{u}{v}$ for any $u\in V$.
\end{invariant}

In order to maintain Invariant~\ref{inv:eliminated} the procedures
eliminate incoming arcs when a vertex is being visited. To do
this efficiently, each vertex $v\in V$ needs an array storing the
vertices $u\in V$ for which there is an arc $\arcg{u}{v}$, as well as
the index $i$ of $v$ in the array $\ADJ[u].\outv$ such that
$v=\ADJ[u].\outv[i]$. The arrays $\ADJ[v].\inv$ for each $v\in V$
shall store the pairs $(u,i)$ representing the incoming arcs of $v$ in
this fashion.

To eliminate the arc $\arcg{u}{v}$ from the adjacency array of $u$,
links (indices) to next and previous non-removed vertices in the array
are maintained, imposing a doubly linked list on each of the
adjacency arrays. The operation $\eliminate(G,u,i)$ removes
the $i$th vertex in $\ADJ[u].\outv$ by linking it out of the doubly
linked list. The adjacency array itself is not changed. Next and
previous indices are maintained in the arrays $\ADJ[u].\nexti$ and
$\ADJ[u].\previ$; $\ADJ[u].\firsti$ shall index the first
non-eliminated vertex in $\ADJ[u].\outv$.

Algorithm~\ref{alg:reverse} shows how to compute the array of incoming
arcs and the pointers for the doubly linked adjacency lists. A
\textbf{for}-construct indicates sequential execution for all values
in some index set (in some order), whereas the \textbf{par}-construct
indicates that the computations for each element in the index set can
be performed in parallel by the available processors. All processors
are assumed to have access to the same memory, and concurrent reading
is allowed. Synchronization is implied at the end of each
\textbf{par}-construct.

\begin{algorithm}
\caption{Computing incoming arcs for each $v\in V$ and doubly linked
  adjacency lists.}
\label{alg:reverse}
\begin{algorithmic}
\PAR{$u\in V$}
\STATE $\ADJ[u].\indeg\gets 0$, 
\STATE $\ADJ[u].\firsti\gets 0$, 
\ENDPAR
\FOR{$u\in V$}
\PAR{$0\leq i<\ADJ[u].\outdeg$}
\STATE $v\gets \ADJ[u].\outv[i]$
\STATE $d\gets \ADJ[v].\indeg$
\STATE $\ADJ[v].\inv[d]\gets (u,i)$ \COMMENT{Add incoming arc $\arcg{u}{v}$ to $v$}
\STATE $\ADJ[v].\indeg\gets d+1$
\STATE $\ADJ[u].\nexti[i]\gets i+1$
\STATE $\ADJ[u].\previ[i]\gets i-1$
\ENDPAR
\ENDFOR
\end{algorithmic}
\end{algorithm}

\begin{lemma}
Algorithm~\ref{alg:reverse} computes the array of $(u,i)$ vertex-index
pairs representing the incoming arcs for all vertices $v\in V$. It
also initializes the doubly linked lists over the adjacency arrays
$\ADJ[u].\outv$.  The algorithm runs in $O(m/p+n)$ time steps with $p$
processors using $O(m+n)$ additional space.
\end{lemma}

\begin{proof}
In each sequential iteration over the set of vertices, incoming arcs
are added to \emph{different} target vertices. For each $u\in V$ this
can therefore be done in parallel by the $p$ available processors in
$O(d(u)/p)$ time steps, where $d(u)$ is the outdegree of vertex
$u$, provided that all processors can read the start address of the array. 
The total time is $O(n+\sum_{u\in V}d(u)/p)=O(m/p+n)$.
\end{proof}

The $\ADJ[u].\firsti$ indices for each $u\in V$ will be maintained
such that $\ADJ[u].\firsti<\ADJ[u].\outdeg$ indicates a non-empty list
of non-eliminated arcs out of $u$. The \eliminate\ operation is
straightforward. 

\section{Depth-First Search}

We can now present the parallel Depth-First Search algorithm. The DFS
procedure is called with a start vertex $s\in V$ and its DFS
number $a$, and computes a DFS tree with reachable vertices numbered
successively in DFS order starting from $a$.  Each recursive call
visits a new vertex, assigns it a DFS number, establishes
Invariant~\ref{inv:eliminated} by eliminating, in parallel, all arcs
into the vertex, and then recursively DFS numbers the subtree from the
first non-eliminated arc $\arcg{s}{v}$ out of $s$; outgoing arcs are
always considered in the fixed order as given in the adjacency array
representation of $G$. The recursion traverses the vertices in $G$
reachable from $s$ in DFS order.  The algorithm is given in detail as
Algorithm~\ref{alg:dfs}, and is essentially as described by Varman and
Doshi~\cite{VarmanDoshi87}.

\begin{algorithm}
\caption{Recursive, parallel Depth-First Search from start vertex
  $s\in V$. Vertices that are reachable from $s$ will be assigned
  successive DFS numbers starting from $a$.}
\label{alg:dfs}
\begin{algorithmic}
\STATE\textbf{Procedure} $\textsf{DFS}(s,G,a)$:
\PAR{$0\leq i<\ADJ[s].\indeg$}
\STATE $(u,j)\gets \ADJ[s].\inv[i]$
\STATE $\eliminate(G,u,j)$
\ENDPAR
\STATE $\ADJ[s].\dfsnum\gets a$ \COMMENT{Vertex $s$ now visited}
\STATE $a\gets a+1$
\WHILE[As long as there are un-eliminated arcs]{$\ADJ[s].\firsti<\ADJ[s].\outdeg$}
\STATE $i\gets\ADJ[s].\firsti$
\STATE $v\gets\ADJ[s].\outv[i]$
\STATE $\ADJ[v].\parent\gets s$
\STATE $a\gets \textsf{DFS}(v,G,a)$
\ENDWHILE
\RETURN $a$
\end{algorithmic}
\end{algorithm}

\begin{proposition}
Algorithm~\ref{alg:dfs} computes an ordered Depth-First Search
numbering and tree in $O(m/p+n)$ time steps using $p$ processors.
\end{proposition}

\begin{proof}
By Invariant~\ref{inv:eliminated} once a vertex $v$ is visited, it
will never be considered again, since all arcs into $v$ will have been
eliminated. Therefore, each vertex in $G$ that is reachable from $s$
will be visited once. The time complexity is immediate: when a vertex
is visited the incoming arcs are eliminated in parallel. Since
$\ADJ[u].\firsti$ will for each vertex be the index of the first
adjacent vertex $v$ where the arc $\arcg{u}{v}$ has not been
eliminated, the order in which vertices are visited is the same as
standard DFS search procedures, from which the correctness follows.
\end{proof}

The algorithm also computes a DFS tree by setting parent pointers for
the visited vertices. Note that it can easily be extended to classify
arcs into backwards, forwards, tree and cross arcs, as
sometimes desirable by a DFS traversal, without changing the time
bounds. An example execution of the
algorithm is given in Figure~\ref{fig:dfs}.

\begin{figure}
\begin{center}
\begin{tabular}{ccccc}
\begin{tikzpicture}
  [scale=.7,auto=left,every node/.style={circle,draw,fill=none}]

  \foreach \pos/\name in {
    {(1,3)/n0},
    {(0,2)/n1}, {(1,2)/n2}, {(2,2)/n3}, {(3,2)/n4},
    {(1,1)/n5}, {(2,1)/n6},
    {(2,0)/n7}, {(3,0)/n8}}{
    \node (\name) at \pos {};
  }

  \foreach \from/\to in {n0/n1,n0/n2,n0/n3,n0/n4,n1/n5,n2/n5,n3/n5}{
    \draw[<->] (\from) -- (\to);
  }
  \foreach \from/\to in {n2/n3,n4/n3,n2/n6,n3/n6,n4/n6,n5/n7,n7/n6,n7/n8,n8/n3,n8/n4}{
    \draw[->] (\from) -- (\to);
  }
\end{tikzpicture}
&
\begin{tikzpicture}
  [scale=.7,auto=left,every node/.style={circle,draw,fill=none}]

  \node[font=\tiny,fill=gray] (n0) at (1,3) {$0$};
  \foreach \pos/\name in {
    {(0,2)/n1}, {(1,2)/n2}, {(2,2)/n3}, {(3,2)/n4},
    {(1,1)/n5}, {(2,1)/n6},
    {(2,0)/n7}, {(3,0)/n8}}{
    \node (\name) at \pos {};
  }

  \foreach \from/\to in {n1/n5,n2/n5,n3/n5}{
    \draw[<->] (\from) -- (\to);
  }
  \foreach \from/\to in {n0/n1,n0/n2,n0/n3,n0/n4,n2/n3,n4/n3,n2/n6,n3/n6,n4/n6,n5/n7,n7/n6,n7/n8,n8/n3,n8/n4}{
    \draw[->] (\from) -- (\to);
  }
\end{tikzpicture}
&
\begin{tikzpicture}
  [scale=.7,auto=left,every node/.style={circle,draw,fill=none}]

  \node[font=\tiny,fill=gray] (n0) at (1,3) {$0$};
  \node[font=\tiny,fill=gray] (n1) at (0,2) {$1$};
  \foreach \pos/\name in {
    {(1,2)/n2}, {(2,2)/n3}, {(3,2)/n4},
    {(1,1)/n5}, {(2,1)/n6},
    {(2,0)/n7}, {(3,0)/n8}}{
    \node (\name) at \pos {};
  }

  \foreach \from/\to in {n2/n5,n3/n5}{
    \draw[<->] (\from) -- (\to);
  }
  \foreach \from/\to in {n1/n5,n0/n2,n0/n3,n0/n4,n2/n3,n4/n3,n2/n6,n3/n6,n4/n6,n5/n7,n7/n6,n7/n8,n8/n3,n8/n4}{
    \draw[->] (\from) -- (\to);
  }
  \draw[thick] (n0) -- (n1);
\end{tikzpicture}
&
\begin{tikzpicture}
  [scale=.7,auto=left,every node/.style={circle,draw,fill=none}]

  \node[font=\tiny,fill=gray] (n0) at (1,3) {$0$};
  \node[font=\tiny,fill=gray] (n1) at (0,2) {$1$};
  \node[font=\tiny,fill=gray] (n5) at (1,1) {$2$};
  \foreach \pos/\name in {
    {(1,2)/n2}, {(2,2)/n3}, {(3,2)/n4},
    {(2,1)/n6},
    {(2,0)/n7}, {(3,0)/n8}}{
    \node (\name) at \pos {};
  }

  \foreach \from/\to in {n5/n2,n5/n3,n0/n2,n0/n3,n0/n4,n2/n3,n4/n3,n2/n6,n3/n6,n4/n6,n5/n7,n7/n6,n7/n8,n8/n3,n8/n4}{
    \draw[->] (\from) -- (\to);
  }
  \draw[thick] (n0) -- (n1);
  \draw[thick] (n1) -- (n5);
\end{tikzpicture}
&
\begin{tikzpicture}
  [scale=.7,auto=left,every node/.style={circle,draw,fill=none}]

  \node[font=\tiny,fill=gray] (n0) at (1,3) {$0$};
  \node[font=\tiny,fill=gray] (n1) at (0,2) {$1$};
  \node[font=\tiny,fill=gray] (n5) at (1,1) {$2$};
  \node[font=\tiny,fill=gray] (n7) at (2,0) {$3$};
  \foreach \pos/\name in {
    {(1,2)/n2}, {(2,2)/n3}, {(3,2)/n4},
    {(2,1)/n6},
    {(3,0)/n8}}{
    \node (\name) at \pos {};
  }

  \foreach \from/\to in {n5/n2,n5/n3,n0/n2,n0/n3,n0/n4,n2/n3,n4/n3,n2/n6,n3/n6,n4/n6,n7/n6,n7/n8,n8/n3,n8/n4}{
    \draw[->] (\from) -- (\to);
  }
  \draw[thick] (n0) -- (n1);
  \draw[thick] (n1) -- (n5);
  \draw[thick] (n5) -- (n7);
\end{tikzpicture}
\\
\begin{tikzpicture}
  [scale=.7,auto=left,every node/.style={circle,draw,fill=none}]

  \node[font=\tiny,fill=gray] (n0) at (1,3) {$0$};
  \node[font=\tiny,fill=gray] (n1) at (0,2) {$1$};
  \node[font=\tiny,fill=gray] (n5) at (1,1) {$2$};
  \node[font=\tiny,fill=gray] (n7) at (2,0) {$3$};
  \node[font=\tiny,fill=gray] (n8) at (3,0) {$4$};
  \foreach \pos/\name in {
    {(1,2)/n2}, {(2,2)/n3}, {(3,2)/n4},
    {(2,1)/n6}}{
    \node (\name) at \pos {};
  }

  \foreach \from/\to in {n5/n2,n5/n3,n0/n2,n0/n3,n0/n4,n2/n3,n4/n3,n2/n6,n3/n6,n4/n6,n7/n6,n8/n3,n8/n4}{
    \draw[->] (\from) -- (\to);
  }
  \draw[thick] (n0) -- (n1);
  \draw[thick] (n1) -- (n5);
  \draw[thick] (n5) -- (n7);
  \draw[thick] (n7) -- (n8);
\end{tikzpicture}
&
\begin{tikzpicture}
  [scale=.7,auto=left,every node/.style={circle,draw,fill=none}]

  \node[font=\tiny,fill=gray] (n0) at (1,3) {$0$};
  \node[font=\tiny,fill=gray] (n1) at (0,2) {$1$};
  \node[font=\tiny,fill=gray] (n5) at (1,1) {$2$};
  \node[font=\tiny,fill=gray] (n7) at (2,0) {$3$};
  \node[font=\tiny,fill=gray] (n8) at (3,0) {$4$};
  \node[font=\tiny,fill=gray] (n4) at (3,2) {$5$};
  \foreach \pos/\name in {
    {(1,2)/n2}, {(2,2)/n3},
    {(2,1)/n6}}{
    \node (\name) at \pos {};
  }

  \foreach \from/\to in {n5/n2,n5/n3,n0/n2,n0/n3,n2/n3,n4/n3,n2/n6,n3/n6,n4/n6,n7/n6,n8/n3}{
    \draw[->] (\from) -- (\to);
  }
  \draw[thick] (n0) -- (n1);
  \draw[thick] (n1) -- (n5);
  \draw[thick] (n5) -- (n7);
  \draw[thick] (n7) -- (n8);
  \draw[thick] (n8) -- (n4);
\end{tikzpicture}
&
\begin{tikzpicture}
  [scale=.7,auto=left,every node/.style={circle,draw,fill=none}]

  \node[font=\tiny,fill=gray] (n0) at (1,3) {$0$};
  \node[font=\tiny,fill=gray] (n1) at (0,2) {$1$};
  \node[font=\tiny,fill=gray] (n5) at (1,1) {$2$};
  \node[font=\tiny,fill=gray] (n7) at (2,0) {$3$};
  \node[font=\tiny,fill=gray] (n8) at (3,0) {$4$};
  \node[font=\tiny,fill=gray] (n4) at (3,2) {$5$};
  \node[font=\tiny,fill=gray] (n3) at (2,2) {$6$};
  \foreach \pos/\name in {
    {(1,2)/n2},
    {(2,1)/n6}}{
    \node (\name) at \pos {};
  }

  \foreach \from/\to in {n5/n2,n0/n2,n2/n6,n3/n6,n4/n6,n7/n6}{
    \draw[->] (\from) -- (\to);
  }
  \draw[thick] (n0) -- (n1);
  \draw[thick] (n1) -- (n5);
  \draw[thick] (n5) -- (n7);
  \draw[thick] (n7) -- (n8);
  \draw[thick] (n8) -- (n4);
  \draw[thick] (n4) -- (n3);
\end{tikzpicture}
&
\begin{tikzpicture}
  [scale=.7,auto=left,every node/.style={circle,draw,fill=none}]

  \node[font=\tiny,fill=gray] (n0) at (1,3) {$0$};
  \node[font=\tiny,fill=gray] (n1) at (0,2) {$1$};
  \node[font=\tiny,fill=gray] (n5) at (1,1) {$2$};
  \node[font=\tiny,fill=gray] (n7) at (2,0) {$3$};
  \node[font=\tiny,fill=gray] (n8) at (3,0) {$4$};
  \node[font=\tiny,fill=gray] (n4) at (3,2) {$5$};
  \node[font=\tiny,fill=gray] (n3) at (2,2) {$6$};
  \node[font=\tiny,fill=gray] (n6) at (2,1) {$7$};
  \foreach \pos/\name in {
    {(1,2)/n2}}{
    \node (\name) at \pos {};
  }

  \foreach \from/\to in {n5/n2,n0/n2}{
    \draw[->] (\from) -- (\to);
  }
  \draw[thick] (n0) -- (n1);
  \draw[thick] (n1) -- (n5);
  \draw[thick] (n5) -- (n7);
  \draw[thick] (n7) -- (n8);
  \draw[thick] (n8) -- (n4);
  \draw[thick] (n4) -- (n3);
  \draw[thick] (n3) -- (n6);
\end{tikzpicture}
&
\begin{tikzpicture}
  [scale=.7,auto=left,every node/.style={circle,draw,fill=none}]

  \node[font=\tiny,fill=gray] (n0) at (1,3) {$0$};
  \node[font=\tiny,fill=gray] (n1) at (0,2) {$1$};
  \node[font=\tiny,fill=gray] (n5) at (1,1) {$2$};
  \node[font=\tiny,fill=gray] (n7) at (2,0) {$3$};
  \node[font=\tiny,fill=gray] (n8) at (3,0) {$4$};
  \node[font=\tiny,fill=gray] (n4) at (3,2) {$5$};
  \node[font=\tiny,fill=gray] (n3) at (2,2) {$6$};
  \node[font=\tiny,fill=gray] (n6) at (2,1) {$7$};
  \node[font=\tiny,fill=gray] (n2) at (1,2) {$8$};

  \draw[thick] (n0) -- (n1);
  \draw[thick] (n1) -- (n5);
  \draw[thick] (n5) -- (n7);
  \draw[thick] (n7) -- (n8);
  \draw[thick] (n8) -- (n4);
  \draw[thick] (n4) -- (n3);
  \draw[thick] (n3) -- (n6);
  \draw[thick] (n5) -- (n2);
\end{tikzpicture}

\end{tabular}
\end{center}
\caption{A sample graph $G=(V,E)$ and the DFS traversal as per
  Algorithm~\ref{alg:dfs} starting from the topmost node. Arcs are
  examined in counter-clockwise order, starting from lower left. Node labels are the DFS numbers, and
  tree edges are indicated as heavy, undirected edges. Arcs disappear
  as they are being eliminated, leaving at the end the heavy DFS tree.}
\label{fig:dfs}
\end{figure}
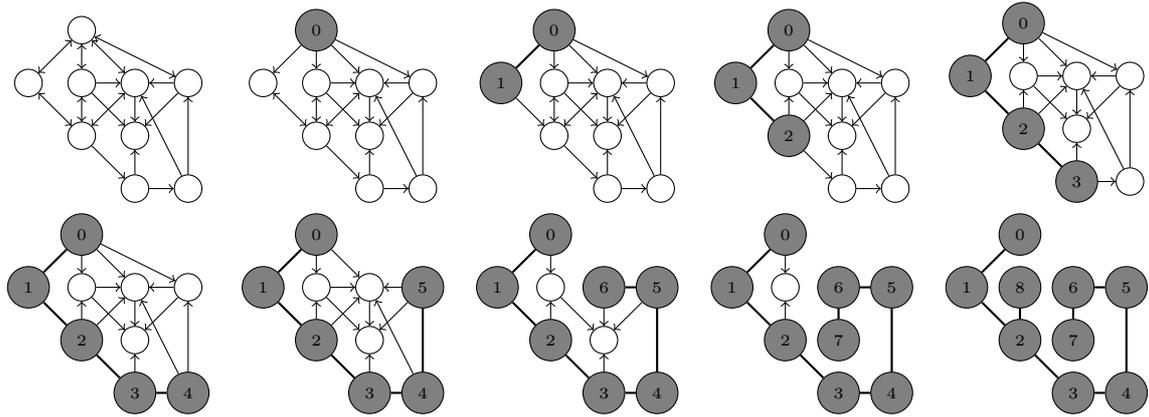

\section{Breadth-First Search}

The arc elimination idea can also be used for parallel Breadth-First
Search, as shown in Algorithm~\ref{alg:bfs}. The algorithm has the
same structure as standard, ``forwards''
BFS~\cite{CormenLeisersonRivestStein09}, but performs parallel arc
elimination as new, unexplored vertices are added to the queue for the
next level. This ensures that each reachable vertex is explored once.

\begin{algorithm}
\caption{Parallel Breadth-First Search from start vertex $s\in
  Q$. Vertices that are reachable from $s$ will be assigned a BFS
  number starting from $a$; also the distance from $s$ (in smallest
  number of arcs) will be computed.}
\label{alg:bfs}
\begin{algorithmic}
\STATE\textbf{Procedure} $\mbox{BFS}(s,G,a)$:
\PAR{$0\leq i<\ADJ[s].\indeg$}
\STATE $(u,j)\gets\ADJ[s].\inv[i]$
\STATE $\eliminate(G,u,j)$
\ENDPAR
\STATE $l\gets 0$
\STATE $\ADJ[s].\bfsnum\gets a$
\STATE $\ADJ[u].\bfslev\gets l$
\STATE $Q.\enque(s)$ \COMMENT{Start vertex $s$ visited}
\STATE $Q'\gets \emptyset$
\REPEAT
\STATE $l\gets l+1$ \COMMENT{Next level}
\REPEAT
\STATE $u\gets Q.\deque()$
\WHILE[As long as there are un-eliminated arcs]{$\ADJ[u].\firsti<\ADJ[u].\outdeg$}
\STATE $i\gets \ADJ[u].\firsti$
\STATE $v\gets\ADJ[u].\outv[i]$ 
\PAR{$0\leq j<\ADJ[v].\indeg$}
\STATE $(w,k)\gets\ADJ[v].\inv[j]$
\STATE $\eliminate(G,w,k)$
\ENDPAR
\STATE $a\gets a+1$ \COMMENT{Next vertex}
\STATE $\ADJ[v].\bfsnum \gets a$
\STATE $\ADJ[v].\bfslev \gets l$
\STATE $\ADJ[v].\parent \gets u$
\STATE $Q'.\enque(v)$ \COMMENT{Vertex $v$ has now been visited, enque for next level}
\ENDWHILE
\UNTIL{$Q=\emptyset$}
\STATE $Q\gets Q'$
\UNTIL{$Q=\emptyset$}
\end{algorithmic}
\end{algorithm}

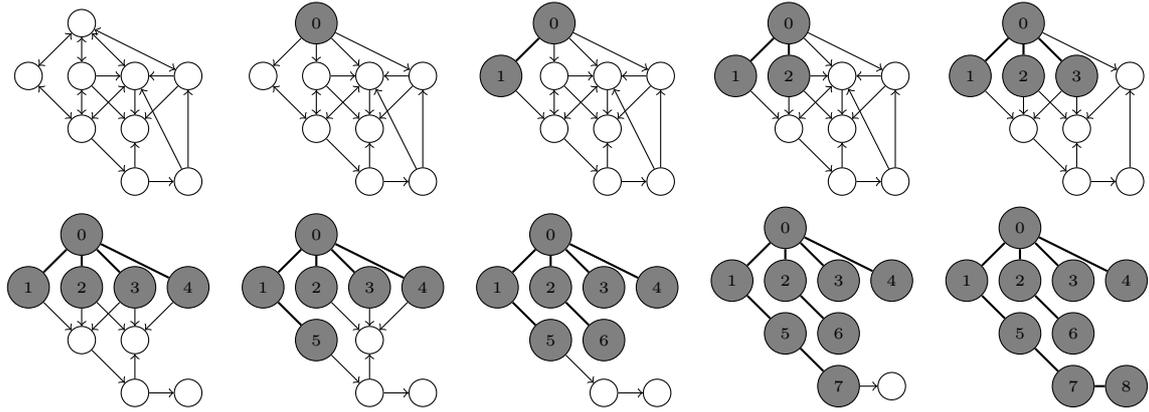
\begin{figure}
\begin{center}
\begin{tabular}{ccccc}
\begin{tikzpicture}
  [scale=.7,auto=left,every node/.style={circle,draw,fill=none}]

  \foreach \pos/\name in {
    {(1,3)/n0},
    {(0,2)/n1}, {(1,2)/n2}, {(2,2)/n3}, {(3,2)/n4},
    {(1,1)/n5}, {(2,1)/n6},
    {(2,0)/n7}, {(3,0)/n8}}{
    \node (\name) at \pos {};
  }

  \foreach \from/\to in {n0/n1,n0/n2,n0/n3,n0/n4,n1/n5,n2/n5,n3/n5}{
    \draw[<->] (\from) -- (\to);
  }
  \foreach \from/\to in {n2/n3,n4/n3,n2/n6,n3/n6,n4/n6,n5/n7,n7/n6,n7/n8,n8/n3,n8/n4}{
    \draw[->] (\from) -- (\to);
  }
\end{tikzpicture}
&
\begin{tikzpicture}
  [scale=.7,auto=left,every node/.style={circle,draw,fill=none}]

  \node[font=\tiny,fill=gray] (n0) at (1,3) {$0$};
  \foreach \pos/\name in {
    {(0,2)/n1}, {(1,2)/n2}, {(2,2)/n3}, {(3,2)/n4},
    {(1,1)/n5}, {(2,1)/n6},
    {(2,0)/n7}, {(3,0)/n8}}{
    \node (\name) at \pos {};
  }

  \foreach \from/\to in {n1/n5,n2/n5,n3/n5}{
    \draw[<->] (\from) -- (\to);
  }
  \foreach \from/\to in {n0/n1,n0/n2,n0/n3,n0/n4,n2/n3,n4/n3,n2/n6,n3/n6,n4/n6,n5/n7,n7/n6,n7/n8,n8/n3,n8/n4}{
    \draw[->] (\from) -- (\to);
  }
\end{tikzpicture}
&
\begin{tikzpicture}
  [scale=.7,auto=left,every node/.style={circle,draw,fill=none}]

  \node[font=\tiny,fill=gray] (n0) at (1,3) {$0$};
  \node[font=\tiny,fill=gray] (n1) at (0,2) {$1$};
  \foreach \pos/\name in {
    {(1,2)/n2}, {(2,2)/n3}, {(3,2)/n4},
    {(1,1)/n5}, {(2,1)/n6},
    {(2,0)/n7}, {(3,0)/n8}}{
    \node (\name) at \pos {};
  }

  \foreach \from/\to in {n2/n5,n3/n5}{
    \draw[<->] (\from) -- (\to);
  }
  \foreach \from/\to in {n1/n5,n0/n2,n0/n3,n0/n4,n2/n3,n4/n3,n2/n6,n3/n6,n4/n6,n5/n7,n7/n6,n7/n8,n8/n3,n8/n4}{
    \draw[->] (\from) -- (\to);
  }
  \draw[thick] (n0) -- (n1);
\end{tikzpicture}
&
\begin{tikzpicture}
  [scale=.7,auto=left,every node/.style={circle,draw,fill=none}]

  \node[font=\tiny,fill=gray] (n0) at (1,3) {$0$};
  \node[font=\tiny,fill=gray] (n1) at (0,2) {$1$};
  \node[font=\tiny,fill=gray] (n2) at (1,2) {$2$};
  \foreach \pos/\name in {
    {(2,2)/n3}, {(3,2)/n4},
    {(1,1)/n5}, {(2,1)/n6},
    {(2,0)/n7}, {(3,0)/n8}}{
    \node (\name) at \pos {};
  }

  \foreach \from/\to in {n3/n5}{
    \draw[<->] (\from) -- (\to);
  }
  \foreach \from/\to in {n1/n5,n0/n3,n0/n4,n2/n3,n4/n3,n2/n6,n3/n6,n4/n6,n2/n5,n5/n7,n7/n6,n7/n8,n8/n3,n8/n4}{
    \draw[->] (\from) -- (\to);
  }
  \draw[thick] (n0) -- (n1);
  \draw[thick] (n0) -- (n2);
\end{tikzpicture}
&
\begin{tikzpicture}
  [scale=.7,auto=left,every node/.style={circle,draw,fill=none}]

  \node[font=\tiny,fill=gray] (n0) at (1,3) {$0$};
  \node[font=\tiny,fill=gray] (n1) at (0,2) {$1$};
  \node[font=\tiny,fill=gray] (n2) at (1,2) {$2$};
  \node[font=\tiny,fill=gray] (n3) at (2,2) {$3$};
  \foreach \pos/\name in {
    {(3,2)/n4},
    {(1,1)/n5}, {(2,1)/n6},
    {(2,0)/n7}, {(3,0)/n8}}{
    \node (\name) at \pos {};
  }

  \foreach \from/\to in {n1/n5,n0/n4,n2/n6,n3/n5,n3/n6,n4/n6,n2/n5,n5/n7,n7/n6,n7/n8,n8/n4}{
    \draw[->] (\from) -- (\to);
  }
  \draw[thick] (n0) -- (n1);
  \draw[thick] (n0) -- (n2);
  \draw[thick] (n0) -- (n3);
\end{tikzpicture}
\\
\begin{tikzpicture}
  [scale=.7,auto=left,every node/.style={circle,draw,fill=none}]

  \node[font=\tiny,fill=gray] (n0) at (1,3) {$0$};
  \node[font=\tiny,fill=gray] (n1) at (0,2) {$1$};
  \node[font=\tiny,fill=gray] (n2) at (1,2) {$2$};
  \node[font=\tiny,fill=gray] (n3) at (2,2) {$3$};
  \node[font=\tiny,fill=gray] (n4) at (3,2) {$4$};
  \foreach \pos/\name in {
    {(1,1)/n5}, {(2,1)/n6},
    {(2,0)/n7}, {(3,0)/n8}}{
    \node (\name) at \pos {};
  }

  \foreach \from/\to in {n1/n5,n2/n6,n3/n5,n3/n6,n4/n6,n2/n5,n5/n7,n7/n6,n7/n8}{
    \draw[->] (\from) -- (\to);
  }
  \draw[thick] (n0) -- (n1);
  \draw[thick] (n0) -- (n2);
  \draw[thick] (n0) -- (n3);
  \draw[thick] (n0) -- (n4);
\end{tikzpicture}
&
\begin{tikzpicture}
  [scale=.7,auto=left,every node/.style={circle,draw,fill=none}]

  \node[font=\tiny,fill=gray] (n0) at (1,3) {$0$};
  \node[font=\tiny,fill=gray] (n1) at (0,2) {$1$};
  \node[font=\tiny,fill=gray] (n2) at (1,2) {$2$};
  \node[font=\tiny,fill=gray] (n3) at (2,2) {$3$};
  \node[font=\tiny,fill=gray] (n4) at (3,2) {$4$};
  \node[font=\tiny,fill=gray] (n5) at (1,1) {$5$};
  \foreach \pos/\name in {
    {(2,1)/n6},
    {(2,0)/n7}, {(3,0)/n8}}{
    \node (\name) at \pos {};
  }

  \foreach \from/\to in {n2/n6,n3/n6,n4/n6,n5/n7,n7/n6,n7/n8}{
    \draw[->] (\from) -- (\to);
  }
  \draw[thick] (n0) -- (n1);
  \draw[thick] (n0) -- (n2);
  \draw[thick] (n0) -- (n3);
  \draw[thick] (n0) -- (n4);
  \draw[thick] (n1) -- (n5);
\end{tikzpicture}
&
\begin{tikzpicture}
  [scale=.7,auto=left,every node/.style={circle,draw,fill=none}]

  \node[font=\tiny,fill=gray] (n0) at (1,3) {$0$};
  \node[font=\tiny,fill=gray] (n1) at (0,2) {$1$};
  \node[font=\tiny,fill=gray] (n2) at (1,2) {$2$};
  \node[font=\tiny,fill=gray] (n3) at (2,2) {$3$};
  \node[font=\tiny,fill=gray] (n4) at (3,2) {$4$};
  \node[font=\tiny,fill=gray] (n5) at (1,1) {$5$};
  \node[font=\tiny,fill=gray] (n6) at (2,1) {$6$};
  \foreach \pos/\name in {
    {(2,0)/n7}, {(3,0)/n8}}{
    \node (\name) at \pos {};
  }

  \foreach \from/\to in {n5/n7,n7/n8}{
    \draw[->] (\from) -- (\to);
  }
  \draw[thick] (n0) -- (n1);
  \draw[thick] (n0) -- (n2);
  \draw[thick] (n0) -- (n3);
  \draw[thick] (n0) -- (n4);
  \draw[thick] (n1) -- (n5);
  \draw[thick] (n2) -- (n6);
\end{tikzpicture}
&
\begin{tikzpicture}
  [scale=.7,auto=left,every node/.style={circle,draw,fill=none}]

  \node[font=\tiny,fill=gray] (n0) at (1,3) {$0$};
  \node[font=\tiny,fill=gray] (n1) at (0,2) {$1$};
  \node[font=\tiny,fill=gray] (n2) at (1,2) {$2$};
  \node[font=\tiny,fill=gray] (n3) at (2,2) {$3$};
  \node[font=\tiny,fill=gray] (n4) at (3,2) {$4$};
  \node[font=\tiny,fill=gray] (n5) at (1,1) {$5$};
  \node[font=\tiny,fill=gray] (n6) at (2,1) {$6$};
  \node[font=\tiny,fill=gray] (n7) at (2,0) {$7$};
  \foreach \pos/\name in {
    {(3,0)/n8}}{
    \node (\name) at \pos {};
  }

  \foreach \from/\to in {n7/n8}{
    \draw[->] (\from) -- (\to);
  }
  \draw[thick] (n0) -- (n1);
  \draw[thick] (n0) -- (n2);
  \draw[thick] (n0) -- (n3);
  \draw[thick] (n0) -- (n4);
  \draw[thick] (n1) -- (n5);
  \draw[thick] (n2) -- (n6);
  \draw[thick] (n5) -- (n7);
\end{tikzpicture}
&
\begin{tikzpicture}
  [scale=.7,auto=left,every node/.style={circle,draw,fill=none}]

  \node[font=\tiny,fill=gray] (n0) at (1,3) {$0$};
  \node[font=\tiny,fill=gray] (n1) at (0,2) {$1$};
  \node[font=\tiny,fill=gray] (n2) at (1,2) {$2$};
  \node[font=\tiny,fill=gray] (n3) at (2,2) {$3$};
  \node[font=\tiny,fill=gray] (n4) at (3,2) {$4$};
  \node[font=\tiny,fill=gray] (n5) at (1,1) {$5$};
  \node[font=\tiny,fill=gray] (n6) at (2,1) {$6$};
  \node[font=\tiny,fill=gray] (n7) at (2,0) {$7$};
  \node[font=\tiny,fill=gray] (n8) at (3,0) {$8$};

  \draw[thick] (n0) -- (n1);
  \draw[thick] (n0) -- (n2);
  \draw[thick] (n0) -- (n3);
  \draw[thick] (n0) -- (n4);
  \draw[thick] (n1) -- (n5);
  \draw[thick] (n2) -- (n6);
  \draw[thick] (n5) -- (n7);
  \draw[thick] (n7) -- (n8);
\end{tikzpicture}
\end{tabular}
\end{center}
\caption{The sample graph $G=(V,E)$ and the BFS traversal as per
  Algorithm~\ref{alg:bfs} starting from the topmost node. Arcs are
  examined in counter-clockwise order, starting from lower left. Node labels
  are the computed BFS numbers, and
  tree edges are indicated as heavy, undirected edges. Arcs disappear
  as they are being eliminated, leaving at the end the heavy BFS tree.}
\label{fig:bfs}
\end{figure}

\begin{proposition}
Algorithm~\ref{alg:bfs} computes an ordered Breadth-First Search
numbering and tree in $O(m/p+n)$ time steps using $p$ processors.
\end{proposition}

\begin{proof}
Again by Invariant~\ref{inv:eliminated}, once a vertex has been
visited it will never be considered again, and therefore arc
elimination is performed once for each reachable vertex. From this the
time bound follows. For each vertex in $Q$ for some level, the
un-eliminated arcs are considered in order determined by the
representation of $G$, and as each unvisited vertex is put into the
queue $Q'$ for the next level, all incoming arcs are eliminated. This
in particular ensures that there are no arcs between vertices in $Q'$.
As for standard BFS, all vertices in $Q$ before the start of an
iteration of the innermost repeat loop have the same distance to the
source vertex, from which correctness follows.
\end{proof}

It is especially worth noticing that there are no arcs between nodes
in $Q'$, the queue being filled for the next iteration. An example
execution of the algorithm is given in Figure~\ref{fig:bfs}. The important
property of BFS is that vertices are explored in least recently
visited order; it is therefore, also in the arc elimination algorithm,
possible to dispense with the explicit next level queue $Q'$ and do
with only a single \textbf{repeat}-loop~\cite{Tarjan83}.

\section{Discussion}

The time bounds for both Depth- and Breadth-First Search algorithms
guarantee linear speed-up when $p\leq m/n$, or equivalently $m\geq
pn$; that is, good speed-up is possible for graphs with average degree
larger than the number of processors.  The algorithms presented here
are complementary to the standard, textbook, ``forwards'' procedures
for DFS and BFS~\cite{CormenLeisersonRivestStein09}.  Standard DFS
where arcs are examined only in the forwards direction has no
parallelism; in contrast, the algorithm given here can perform the arc
elimination fully in parallel.  Typical, parallel Breadth-First Search
algorithms exploit parallelism mostly by considering active vertices
in the queue for each level in parallel, see,
e.g.,~\cite{BaderMadduri06,LeisersonSchardl10,ShunBlelloch13}. Although
the forward edges can also be explored in parallel, compaction or
other data structure operations are necessary for resolving/avoiding
update conflicts and maintaining the queue for the next level. The
parallel running time of such algorithms will typically be bounded by
the diameter of the graph, but either at the cost of more work
incurred by data structure operations, or by requiring stronger,
atomic operations.  The arc elimination approach does not require
either of these means (data structures, compaction, atomic
operations), and has exploitable parallelism independent of the BFS
structure of the graphs, as long as the total number of edges $m$
satisfies $m\geq np$.

\bibliographystyle{abbrv}
\bibliography{traff,parallel}

\begin{thebibliography}{1}

\bibitem{BaderMadduri06}
D.~A. Bader and K.~Madduri.
\newblock Designing multithreaded algorithms for breadth-first search and
  $st$-connectivity on the {Cray} {MTA-2}.
\newblock In {\em International Conference on Parallel Processing {(ICPP)}},
  pages 523--530, 2006.

\bibitem{CormenLeisersonRivestStein09}
T.~H. Cormen, C.~E. Leiserson, R.~L. Rivest, and C.~Stein.
\newblock {\em Introduction to Algorithms}.
\newblock {MIT} Press, third edition, 2009.

\bibitem{EdwardsVishkin12}
J.~A. Edwards and U.~Vishkin.
\newblock Better speedups using simpler parallel programming for graph
  connectivity and biconnectivity.
\newblock In {\em International Workshop on Programming Models and Applications
  for Multicores and Manycores ({PMAM})}, pages 103--114, 2012.

\bibitem{LeisersonSchardl10}
C.~E. Leiserson and T.~B. Schardl.
\newblock A work-efficient parallel breadth-first search algorithm (or how to
  cope with the nondeterminism of reducers).
\newblock In {\em 22nd Annual {ACM} Symposium on Parallelism in Algorithms and
  Architectures {(SPAA)}}, pages 303--314, 2010.

\bibitem{Reif85}
J.~H. Reif.
\newblock Depth-first search is inherently sequential.
\newblock {\em {I}nformation {P}rocessing {L}etters}, 20:229--234, 1985.

\bibitem{ShunBlelloch13}
J.~Shun and G.~E. Blelloch.
\newblock Ligra: a lightweight graph processing framework for shared memory.
\newblock In {\em {ACM SIGPLAN} Symposium on Principles and Practice of
  Parallel Programming {(PPoPP)}}, pages 135--146, 2013.

\bibitem{Tarjan72}
R.~E. Tarjan.
\newblock Depth-first search and linear graph algorithms.
\newblock {\em {SIAM} Journal on Computing}, 1(2):146--160, 1972.

\bibitem{Tarjan83}
R.~E. Tarjan.
\newblock {\em Data Structures and Network Algorithms}.
\newblock Society of Industrial and Applied Mathematics ({SIAM}), 1983.

\bibitem{VarmanDoshi87}
P.~J. Varman and K.~Doshi.
\newblock Improved parallel algorithms for the depth-first search and monotone
  circuit value problems.
\newblock In {\em 15th ACM Conference on Computer Science}, pages 175--182,
  1987.

\end{thebibliography}

\end{document}